\documentclass[10pt, journal]{IEEEtran}
\IEEEoverridecommandlockouts

\usepackage{graphicx}
\usepackage{amssymb}
\usepackage{amsmath}
\usepackage{cite}
\usepackage{subfigure}
\usepackage{xcolor}
\usepackage[displaymath,mathlines]{lineno}
\usepackage{tabulary}
\usepackage{multirow}
\usepackage{algorithm}
\usepackage{algpseudocode}
\usepackage{algorithmicx}
\usepackage{multicol}
\usepackage{lipsum}
\usepackage{amsthm}
\usepackage{lipsum}
\usepackage{epstopdf}
\usepackage{enumitem}
\newcommand{\doublewidetilde}[1]{{%
		\mathpalette\double@widetilde{#1}%
}}

\makeatletter

\def\BState{\State\hskip-\ALG@thistlm}
\makeatother

\newtheorem{lemma}{Lemma}

\begin{document}
	%
	\title{Pilot Assignment for Joint Uplink-Downlink Spectral Efficiency Enhancement in Massive MIMO Systems with Spatial Correlation}

	\author{\IEEEauthorblockN{ Tien Hoa Nguyen, Trinh Van Chien,  \textit{Member}, \textit{IEEE}, Hien Quoc Ngo, \textit{Senior Member}, \textit{IEEE}, \\ Xuan~Nam~Tran, \textit{Member}, \textit{IEEE}, and Emil Bj\"{o}rnson, \textit{Senior Member}, \textit{IEEE}}
		\thanks{Manuscript  received  April 25,  2021;    accepted June  14,  2021. The associate editor coordinating the review of this paper and approving it for publication was Dr. J. Joung. (Corresponding author: Trinh Van Chien).
			
			This paper was accepted by IEEE Transactions on Vehicular Technology.  ©2021IEEE. Personal use of this material is permitted. Permission from IEEE must be  obtained  for  all  other  uses,  in  any  current  or  future  media,  including reprinting/republishing this material for advertising or promotional purposes,creating new collective  works, for resale or redistribution to  servers or lists,or reuse of any copyrighted component of this work in other works.
			
			T. H. Nguyen  is with the School of Electronics and Telecommunications (SET), Hanoi University of Science and Technology, Vietnam (email: hoa.nguyentien@hust.edu.vn).
			
			T. V. Chien  was with the School of Information and Communication Technology (SoICT) – Hanoi University of Science and Technology (HUST), Vietnam. He is now with the  Interdisciplinary  Centre  for  Security,  Reliability  and  Trust  (SnT),University  of  Luxembourg,  L-1855  Luxembourg,  Luxembourg (email: trinhchien.dt3@gmail.com).
			
			H. Q. Ngo is with the Institute of Electronics, Communications and Information Technology (ECIT), Queen's University Belfast, Belfast, BT3 9DT, United Kingdom (e-mail: hien.ngo@qub.ac.uk).
			
			X. N. Tran is with Advanced Wireless Communications Group, Le Quy Don Technical University, Vietnam (email: namtx@lqdtu.edu.vn).
			
			E. Bj\"ornson is with the Department of Electrical Engineering, Link\"oping University, 581 83 Link\"oping, Sweden, and with the Department of Computer Science, KTH, 164 40 Kista, Sweden (email: emilbjo@kth.se).
			
			This research is funded by Vietnam National Foundation for Science and Technology Development (NAFOSTED) under grant number 102.01-2019.07.}
	}

	\maketitle
	
	
	\begin{abstract}
		This paper proposes a flexible pilot assignment method to jointly optimize the uplink and downlink data transmission in multi-cell Massive multiple input multiple output (MIMO) systems with correlated Rayleigh fading channels. By utilizing a closed-form expression of the ergodic spectral efficiency (SE) achieved with maximum ratio processing, we formulate an optimization problem for maximizing the minimum weighted sum of the uplink and downlink SEs subject to the transmit powers and pilot assignment sets. This combinatiorial optimization problem is solved by two sequential algorithms: a heuristic pilot assignment is first proposed to obtain a good pilot reuse set and the data power control is then implemented. Numerical results manifest that the proposed algorithm converges fast to a better minimum sum SE per user than the algorithms in previous works. 
	\end{abstract}
	
	
	\IEEEpeerreviewmaketitle
	
	\section{Introduction}
	Massive multiple input multiple output (MIMO) is a key component of 5G since this technology can improve ergodic spectral efficiency (SE) with the orders of magnitude compared to single-antenna systems. It is achieved by equipping each base station (BS) with a large number of antennas such that the system can spatially multiplex tens of users at the same time and frequency resource \cite{Marzetta2016a,le2020pareto}. The quasi-orthogonality of all channels allows a simple linear beamforming to yield SE close to the channel capacity for single-cell systems, in which orthogonal pilot signals are assumed to be available for all users. In cellular Massive MIMO systems, it is an impractical assumption because the pilot overhead is directly proportional to the total number of users \cite{Chien2018a}, while the coherence interval is limited. A small set of orthogonal pilot signals should be reused among the cells, which results in mutually correlated interference known as pilot contamination that downgrades the ergodic SE \cite{Bjornson2017bo}.
	
	In order to mitigate pilot contamination, an observation is that some users will cause more severe contamination to each other when they use the same pilot. This issue should be avoided in the pilot assignment. The system can reuse pilot signals in a way that gives these users a priority to use orthogonal pilots. Nonetheless, an optimal pilot assignment is expensive to find since it is attained by solving a combinatorial problem. Heuristic algorithms with affordable complexity are necessary in practice to eliminate pilot contamination at a reasonable cost. It should be noticed that most previous works only consider the pilot assignment for either the uplink or downlink transmission, see \cite{Chien2018a, Jin2015a} and references therein. The authors \cite{marinello2017joint} were the first to propose a heuristic pilot assignment algorithm taking both uplink and downlink into account, but based on the asymptotic SE from uncorrelated Rayleigh fading that behaves differently from the capacity regime for a limited number of BS antennas and spatially correlated channels. Meanwhile, the combinations of pilot signals to obtain the pilot assignment increases rapidly with the total number of users in all cells making it an intractable solution. Motivated by the fact that uncorrelated channels rarely appear in practice \cite{Gao2015a}, the pilot assignment for spatially correlated channels were considered in \cite{Chien2018a, You2015a}, but the pilot assignment for jointly enhancing the uplink and downlink SE has not been considered in this context.
	
	In this paper, we  assign the pilot signals to jointly maximize the minimum weighted sum of uplink and downlink SE per user with spatially correlated channels. This optimization problem has flexibility since the weights can be used to assign different priorities to the downlink and uplink. The optimization problem is based only on statistical channel information, thus the obtained solution can be utilized as long as the statistics remain the same.  Since this optimization problem is combinatorial and NP-hard, we propose a heuristic pilot assignment that works well for systems with a limited number of BS antennas. The obtained pilot assignment solution outperforms the other benchmarks from previous works.
	
	\textit{Notation}: The upper and lower-bold letters denote vectors and matrices, respectively. The notation $(\cdot)^H$ is the Hermitian transpose. $\mathbb{E} \{ \cdot \}$ is the expectation of a random variable, while $\mathcal{CN}(\cdot,\cdot)$ denotes circularly complex Gaussian distribution. The identity matrix of size $M\times M$ is denoted by $\mathbf{I}_M$. Finally, $\mathrm{tr}(\cdot)$ is the trace operator, $\| \cdot \|$ is Euclidean norm, and $\mathcal{O}(\cdot)$ denotes the big-$\mathcal{O}$ notation.
	
	\section{System Model} \label{Section: System Model}
	We consider a multi-cell Massive MIMO system comprising $L$ cells, each with an $M$-antenna BS serving $K$ single-antenna users. The system uses a time-division duplexing protocol. Let $\tau_c$ be the length of each coherence block whereof $\tau_p$ symbols are used for the uplink training and the remaining is used for the data transmission. We denote by $\gamma^{\mathrm{ul}}$ and $\gamma^{\mathrm{dl}}$ the fraction of the $\tau_c - \tau_p$ symbols used for the uplink and downlink data transmission and satisfied $\gamma^{\mathrm{ul}} + \gamma^{\mathrm{dl}} =1$. The set $\mathcal{S}$ contains all tuple of cell and user indices in the system as
	\begin{equation}
		\mathcal{S} = \left\{ (i,t): i \in \{1, \ldots, L\}, t \in \{1, \ldots, K \} \right\}.
	\end{equation} 
	The channel between user~$t$ in cell~$i$ and BS~$l$ is denoted as $\mathbf{h}_{i,t}^l \in \mathbb{C}^M$ and is assumed to feature correlated Rayleigh fading:
	$\mathbf{h}_{i,t}^l \sim \mathcal{CN} (\mathbf{0}, \mathbf{R}_{i,t}^l)$,
	where $\mathbf{R}_{i,t}^l \in \mathbb{C}^{M \times M}$ is the channel correlation matrix. All BSs know the channel statistics, but need to estimate the realizations in each coherence block.\footnote{For sake of the simplicity, we assume that the channel correlation matrices are known. In practical systems, we can easily estimate the channel correlation matrices by averaging over many different instantaneous channel realizations.}

	\subsection{Uplink training phase}
	Let us introduce a set $\mathcal{P}$ of $\tau_p$ mutually orthogonal pilot signals, $K \leq \tau_p \leq KL,$ reused among the users. The pilot signal assigned to user~$t$ in cell~$l$ is denoted as $\pmb{\psi}_{i,t}$. In the uplink training phase, the received baseband pilot signal $\mathbf{Y}_{p,l} \in \mathbb{C}^{M \times \tau_p}$ at BS~$l$ is
	\begin{equation}\label{eq:Ylk}
		\mathbf{Y}_{p,l} = \sum_{(i,t) \in \mathcal{S} } \mathbf{h}_{i,t}^l \pmb{\psi}_{i,t}^H  + \mathbf{N}_{p,l},
	\end{equation}
	where $\mathbf{N}_{p,l}$ is an $M \times \tau_p$ noise matrix with independent elements distributed as $\mathcal{CN}(0, \sigma_{\textrm{UL}}^2)$  and $\sigma_{\textrm{ul}}^2$ being the noise variance in the uplink. The channel estimate $\hat{\mathbf{h}}_{l,k}^l$ of $\mathbf{h}_{l,k}^l$ is obtained from \eqref{eq:Ylk} by MMSE estimation \cite{Bjornson2017bo} as
	\begin{equation} \label{eq:ChannelEst}
		\begin{split}
			\hat{\mathbf{h}}_{l,k}^l 
			&=  \| \pmb{\psi}_{l,k} \|^2 \mathbf{R}_{l,k}^l (\mathbf{F}_{l,k}^l)^{-1} \mathbf{Y}_{p,l} \pmb{\psi}_{l,k},
		\end{split}
	\end{equation}
	where $\mathbf{F}_{l,k}$ is given as
	\begin{equation}
	\mathbf{F}_{l,k} = \sum_{(i,t) \in \mathcal{S} } \mathbf{R}_{i,t} ^l | \pmb{\psi}_{i,t}^H \pmb{\psi}_{l,k} |^2 +   \sigma_{\mathrm{UL}}^2 \| \pmb{\psi}_{l,k} \|^2 \mathbf{I}_M.
	\end{equation}
	For all $l,k,$ the channel estimates are distributed as
	\begin{equation} \label{eq:DistChannelEst}
		\hat{\mathbf{h}}_{l,k}^l \sim \mathcal{CN} \left( \mathbf{0}, \| \pmb{\psi}_{l,k} \|^4 \mathbf{R}_{l,k}^l \mathbf{F}_{l,k}^{-1} \mathbf{R}_{l,k}^l \right).
	\end{equation}
	The channel estimate in \eqref{eq:ChannelEst} together with its statistical information in \eqref{eq:DistChannelEst} are used to formulate the linear processing vectors for the data transmission and compute closed-form ergodic SE expressions for each user.

	\subsection{Data transmission}
	In the uplink data transmission, the $K$ users in each cell simultaneously send data to the serving BS. Specifically, user $k$ in cell $l$ sends a complex data symbol $s_{l,k}$ with $\mathbb{E} \{ |s_{l,k} |^2\} = 1$. The received signal $\mathbf{y}_l \in \mathbb{C}^M$ at BS~$l$ is
	\begin{equation}
		\mathbf{y}_l = \sum_{(i,t) \in \mathcal{S} } \sqrt{p_{i,t}^{\mathrm{ul}}} \mathbf{h}_{i,t}^l s_{i,t} + \mathbf{n}_l,
	\end{equation}
	where $p_{i,t}^{\mathrm{ul}}$ is the transmit data power and $\mathbf{n}_l \in \mathbb{C}^M$ is the uplink Gaussian noise with $\mathbf{n}_l \sim \mathcal{CN}(\mathbf{0}, \sigma_{\textrm{ul}}^2 \mathbf{I}_M )$. We assume BS~$l$ detects the desired signal from its user~$k$ by utilizing a maximum-ratio combining vector as 
	\begin{equation}
	\mathbf{v}_{l,k} = \hat{\mathbf{h}}_{l,k}^l.
	\end{equation}
	The desired signal is then obtained from
	\begin{equation} \label{eq: DetectedSignal}
		\begin{split}
			& \mathbf{v}_{l,k}^{H} \mathbf{y}_l = \sum_{(i,t) \in \mathcal{S} } \sqrt{p_{i,t}^{\mathrm{ul}}} \hat{\mathbf{h}}_{l,k}^{l,H} \mathbf{h}_{i,t} s_{i,t} + \hat{\mathbf{h}}_{l,k}^{l,H} \mathbf{n}_l.
		\end{split}
	\end{equation}
	
	In the downlink data transmission, BS~$l$ transmits a signal $\mathbf{x}_l \in \mathbb{C}^M$ to its $K$ users, which is formulated as
	\begin{equation}
		\mathbf{x}_l = \sum_{t=1}^K \sqrt{p_{l,t}^{\mathrm{dl}}} \mathbf{w}_{l,t} q_{l,t},
	\end{equation}
	where $p_{l,t}^{\mathrm{dl}}$ is the power allocated to the data symbol $q_{l,t}$ with $\mathbb{E} \{ |q_{l,t}|^2 \} = 1$. The maximum ratio (MR) precoding vector
	\begin{equation} \label{eq:NormalizedMR}
		\mathbf{w}_{l,t} = \frac{\hat{\mathbf{h}}_{l,t}^l}{\sqrt{ \mathbb{E} \left\{ \| \hat{\mathbf{h}}_{l,t}^l \|^2 \right\}}} =  \frac{\hat{\mathbf{h}}_{l,t}^l}{ \sqrt{ \| \pmb{\psi}_{l,k} \|^4 \mathrm{tr} ( \mathbf{R}_{l,k}^l \mathbf{F}_{l,k}^{-1}  \mathbf{R}_{l,k}^l)}},
	\end{equation}
	is used. The received signal at user $k$ in cell $l$ is a superposition of the  signals from all $L$ BSs as
	\begin{equation} \label{eq:DLReceivedSig}
		\begin{split}
			&r_{l,k} =  \sum_{(i,t) \in \mathcal{S} } \sqrt{p_{i,t}^{\mathrm{dl}}}  \left(\mathbf{h}_{l,k}^{i} \right)^H \mathbf{w}_{i,t} q_{i,t} + n_{l,k},
		\end{split}
	\end{equation} 
	where $n_{l,k}$ denotes the additive noise which is distributed as $n_{l,k} \sim \mathcal{CN}(0, \sigma_{\textrm{dl}}^2)$ and $\sigma_{\textrm{dl}}^2$ is the noise variance in the downlink. By applying the standard Massive MIMO methodology \cite{Bjornson2017bo} to \eqref{eq: DetectedSignal} and \eqref{eq:DLReceivedSig}, the closed-form expression of the ergodic uplink and downlink SEs in Lemma~\ref{lemma:ClosedForm} are attained. 
	\begin{lemma}  \label{lemma:ClosedForm}
		Closed-form expression of the uplink and downlink ergodic SEs of user $k$ in cell $l$ are respectively
		\begin{align} 
			R_{l,k}^{\mathrm{ul}}  &= \gamma^{\mathrm{ul}} \left(1  - \frac{\tau_p}{\tau_c} \right) \log_2 \left( 1 + \mathrm{SINR}_{l,k}^{\mathrm{ul}} \right), \label{eq:ClosedULRate}\\
			R_{l,k}^{\mathrm{dl}} &= \gamma^{\mathrm{dl}} \left(1 - \frac{\tau_p}{\tau_c} \right) \log_2 \left(1 + \mathrm{SINR}_{l,k}^{\mathrm{dl}} \right), \label{eq:DLCorrelatedRate}
		\end{align}
		where the effective SINR values are given in \eqref{eq:ULSINR} and \eqref{eq:DLSINR}.
		\begin{figure*}
			\begin{equation} \label{eq:ULSINR}
				\mathrm{SINR}_{l,k}^{\mathrm{ul}} = \frac{p_{l,k}^{\mathrm{ul}} \| \pmb{\psi}_{l,k} \|^4 \mathrm{tr} ( \mathbf{R}_{l,k}^l \mathbf{F}_{l,k}^{-1} \mathbf{R}_{l,k}^l  ) }{ \sum_{(i,t) \in \mathcal{S} \setminus \{(l,k) \}} p_{i,t}^{\mathrm{ul}} | \pmb{\psi}_{l,k}^H \pmb{\psi}_{i,t} |^2 \frac{| \mathrm{tr} \left( \mathbf{R}_{i,t}^l \mathbf{F}_{l,k}^{-1} \mathbf{R}_{l,k}^l \right) |^2}{\mathrm{tr} (\mathbf{R}_{l,k}^{l} \mathbf{F}_{l,k}^{-1}  \mathbf{R}_{l,k}^{l} )} + \sum_{(i,t) \in \mathcal{S} }  p_{i,t}^{\mathrm{ul}} \frac{\mathrm{tr} (  \mathbf{R}_{i,t}^l \mathbf{R}_{l,k}^{l} \mathbf{F}_{l,k}^{-1}  \mathbf{R}_{l,k}^{l} ) }{\mathrm{tr} (\mathbf{R}_{l,k}^{l} \mathbf{F}_{l,k}^{-1}  \mathbf{R}_{l,k}^{l} )} +  \sigma_{\mathrm{ul}}^2}
			\end{equation}
			\begin{equation} \label{eq:DLSINR}
				\mathrm{SINR}_{l,k}^{\mathrm{dl}} = \frac{p_{l,k}^{\mathrm{dl}} \| \pmb{\psi}_{l,k} \|^4  \mathrm{tr} ( \mathbf{R}_{l,k}^l \mathbf{F}_{l,k}^{-1} \mathbf{R}_{l,k}^l ) }{  \sum_{(i,t) \in \mathcal{S} \setminus \{(l,k) \}} p_{i,t}^{\mathrm{dl}} |\pmb{\psi}_{l,k}^H \pmb{\psi}_{i,t}|^2 \frac{  | \mathrm{tr} (  \mathbf{R}_{i,t}^{i} \mathbf{F}_{i,t}^{-1} \mathbf{R}_{l,k}^{i}) |^2  }{\mathrm{tr} \left(\mathbf{R}_{i,t}^i  \mathbf{F}_{i,t}^{-1} \mathbf{R}_{i,t}^i\right) }  + \sum_{(i,t) \in \mathcal{S}} p_{i,t}^{\mathrm{dl}} \frac{   \mathrm{tr} (  \mathbf{R}_{i,t}^{i} \mathbf{F}_{i,t}^{-1} \mathbf{R}_{i,t}^{i} \mathbf{R}_{l,k}^{i}) }{\mathrm{tr} (\mathbf{R}_{i,t}^i  \mathbf{F}_{i,t}^{-1} \mathbf{R}_{i,t}^i ) } + \sigma_{\mathrm{dl}}^2} 
			\end{equation}
			\hrulefill
		\end{figure*}
	\end{lemma}
	\begin{proof}
		The proof follows along the lines of Corollaries 4.5 and 4.9 in \cite{Bjornson2017bo} except for the different notations and the fact that pilot reuse pattern is arbitrary and not defined in advance.
	\end{proof}
In the SINR expressions, the numerator represents an array gain as the trace of the covariance matrix is proportional to the number of BS antennas. The first term of the denominator represents coherent interference originating from pilot contamination caused by the pilot reuse and it grows with the number of BS antennas. The last terms of the denominator are noncoherent interference and noise. While the uplink SINR expression of each user only depends on the channel estimate of this own user, a superposition of the channel estimation quality from all the users are observed in the downlink SINR expression. The denominators of the SINR expressions in \eqref{eq:ULSINR} and \eqref{eq:DLSINR} indicate the different contributions of a pilot reuse pattern to the uplink and downlink data transmission. The coupled nature motivates a pilot assignment for jointly optimizing the both SEs per user instead of either the uplink or downlink SE as in most previous works.
	
\begin{figure}[t]
	\centering
	\includegraphics[trim=0cm 0cm 0cm 0cm, clip=true, width=2.3in]{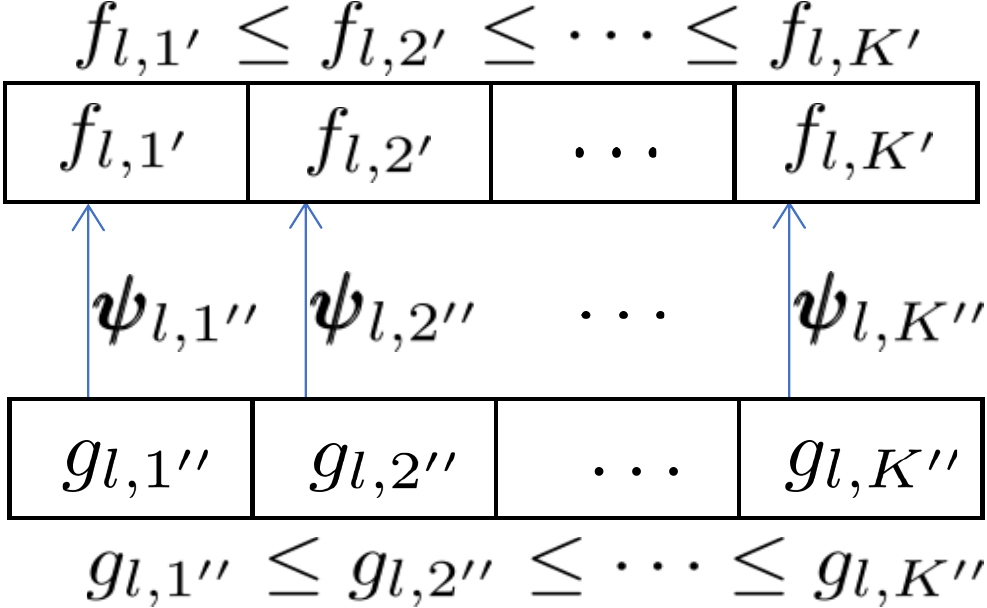} 
	\caption{The proposed pilot assignment for the $K$ users in cell~$l$.}
	\label{FigPilotAssigment}
\end{figure}

\section{Max-Min Fairness Optimization}
This section studies the pilot assignment for the weighted max-min sum SE per user fairness problem with uplink and downlink transmit power constraints. Due to the inherent non-convexity, we propose a heuristic algorithm to obtain a good local solution with tolerable computational complexity.
\subsection{Problem formulation}
By introducing the weights $\{w_{l,k}^{\mathrm{ul}}, w_{l,k}^{\mathrm{dl}} \}$ that prioritize the uplink and downlink transmission of arbitrary user~$k$ in cell~$l$, the optimization problem is formulated for a given set of orthogonal pilot signals as 
\begin{subequations} \label{Problem:DataOptimizationv1}
	\begin{align} 
		\underset{ \substack{ \{ p_{l,k}^{\mathrm{ul}} \geq 0 \},  \{p_{l,k}^{\mathrm{dl}} \geq 0 \}, \\ \{ \pmb{\psi}_{l,k}  \in \mathcal{P} \}  }}{\mathrm{maximize}} &  \quad \underset{(l,k)}{\mathrm{min}} 
		\quad  w_{l,k}^{\mathrm{ul}} R_{l,k}^{\mathrm{ul}} +  w_{l,k}^{\mathrm{dl}} R_{l,k}^{\mathrm{dl}} \\
		\textrm{subject to} & \quad
		p_{l,k}^{\mathrm{ul}} \leq P_{\mathrm{max},l,k}^{\mathrm{ul}} \;, \forall l,k, \label{eq:ULPowerConst}\\
		& \quad \sum_{k=1}^K p_{l,k}^{\mathrm{dl}} \leq P_{\mathrm{max},l}^{\mathrm{dl}} \;, \forall l, \label{eq:DLPowerConst}
	\end{align}
\end{subequations}
where $P_{\mathrm{max},l,k}^{\mathrm{ul}}, P_{\mathrm{max},l}^{\mathrm{dl}}$ is the maximum power that each user and BS can allocate to in the uplink and downlink, respectively. Problem~\eqref{Problem:DataOptimizationv1} is combinatorial and its optimal solution for the pilots is obtained by exhaustive search over possible pilot assignments. For the pilot length of $\tau_p = K$, there are $(K!)^{L-1}$  different pilot assignments, which is impossible to perform in a large-scale system \cite{Chien2018a}. We notice that by introducing weights and considering spatially correlated fading, problem~\eqref{Problem:DataOptimizationv1} is a generalization of previous works which only focus on uncorrelated Rayleigh channel for either the uplink or downlink transmission \cite{Xu2015a}.

\subsection{Heuristic pilot assignment with fixed data powers} \label{SubSec:PilotAssignment}
A low-complexity heuristic pilot assignment algorithm is proposed, in which the user having the lowest weight sum SE is prioritized. For a given set of transmit power coefficients, problem~\eqref{Problem:DataOptimizationv1} becomes
\begin{equation} \label{Problem:PilotAssignment}
	\begin{aligned}  
		& \underset{ \{ \pmb{\psi}_{l,k} \} \in \mathcal{P}  }{\mathrm{maximize}} && \underset{(l,k)}{\mathrm{min}}
		\quad  w_{l,k}^{\mathrm{ul}} R_{l,k}^{\mathrm{ul}} +  w_{l,k}^{\mathrm{dl}} R_{l,k}^{\mathrm{dl}}.
	\end{aligned}
\end{equation}
We assume that all $KL$ users first randomly select the pilot signals such that there is no pilot contamination inside a cell. After that, cell~$l$ reallocates the pilot signals to its $K$ users with the availability of pilot assignment information from other cells. If we define the weighted sum SE of user~$k$ in cell~$l$ as
\begin{equation}
	f_{l,k} =  w_{l,k}^{\mathrm{ul}} R_{l,k}^{\mathrm{ul}} +  w_{l,k}^{\mathrm{dl}} R_{l,k}^{\mathrm{dl}},
\end{equation}
then BS~$l$ can sort all $K$ users in the ascending order as
\begin{equation} \label{eq:Increasingcost}
	f_{l,1'} \leq  f_{l,2'} \leq \ldots  \leq f_{l,K'},
\end{equation}
where $\{ 1', 2', \ldots, K'\}$ is a permutation of the set $\{1, 2, \ldots, K\}$. From these notations, user~$k'$ in cell~$l$ has the weighted sum SE $f_{l,k'}, \forall k' = 1, \ldots, K$.
	
We now compute the normalized mean square error (NMSE) of user $k$ in cell~$l$ as
\begin{equation} \label{eq:PilotContCost}
	\begin{split}
		g_{l,k} &= \frac{ \mathbb{E} \{ \| \mathbf{e}_{l,k}^l \|^2 \} }{\mathbb{E} \{ \| \mathbf{h}_{l,k}^l \|^2 \} } = 1 - \frac{\mathrm{tr} ( \mathbf{R}_{l,k}^l \mathbf{F}_{l,k}^{-1} \mathbf{R}_{l,k}^l  )}{\mathrm{tr} (\mathbf{R}_{l,k}^l )},
	\end{split}
\end{equation}
then the channel estimation quality of the $K$ users in cell~$l$ is sorted in the ascending order as
\begin{equation} \label{eq:PilotContOrder}
	g_{l,1''} \leq g_{l,2''} \leq \ldots  \leq g_{l,K''},
\end{equation}
where $\{ 1'', 2'', \ldots, K''\}$ is a permutation of the set $\{1, 2, \ldots, K\}$. The pilot signal $\pmb{\psi}_{l,k''}$ currently used by user $k''$ is reassigned to user~$k'$. The intuition is to dedicate pilots signal subject to less pilot contamination to users with worse conditions, which is viewed as one with a smaller $f_{l,k'}$. Our proposed pilot assignment for users in cell~$l$ is illustrated in Fig.~\ref{FigPilotAssigment}. This process is implemented cell-by-cell in an iterative algorithm (please see Algorithm~\ref{Algorithmv1}). Concerning on the per-cell-based pilot assignment, a new pilot reuse set may harm the minimum weighted sum SE in the entire system. To avoid this issue, we introduce a backtracking condition to assign the pilot signals at iteration~$n$ as in Lemma~\ref{Lemma:BackTrackingCond}.
\begin{lemma} \label{Lemma:BackTrackingCond}
	If the pilot signals are only assigned to the $K$ users in cell~$l$ when the objective function in \eqref{Problem:PilotAssignment}  does not increase, then the proposed iterative pilot assignment approach converges to a fixed point.
\end{lemma}
\begin{proof}
	Let us denote $h_{l}^{\ast, (n)}$ and $h_{l}^{\ast, (n-1)}$  as the minimum weighted sum SE per user before and after BS~$l$ reassigns the pilot signals, i.e.,
	\begin{align}
		&h_l^{\ast, (n)} = \underset{(l',k)}{\mathrm{min}}
		\quad  w_{l',k}^{\mathrm{ul}} R_{l',k}^{\mathrm{ul}, (n)} +  w_{l',k}^{\mathrm{dl}} R_{l',k}^{\mathrm{dl}, (n)}, \label{eq:hln1}\\
		&h_l^{\ast, (n-1)} = \underset{(l',k)}{\mathrm{min}}
		\quad  w_{l',k}^{\mathrm{ul}} R_{l',k}^{\mathrm{ul},(n-1)} +  w_{l',k}^{\mathrm{dl}} R_{l',k}^{\mathrm{dl}, (n-1)}, \label{eq:hln2}
	\end{align}
	where $R_{l',k}^{\mathrm{ul}, (n)}, R_{l',k}^{\mathrm{dl}, (n)}, R_{l',k}^{\mathrm{ul}, (n-1)},$ and $R_{l',k}^{\mathrm{dl}, (n-1)}$ are the uplink and downlink SEs at iteration~$n$ and $n-1$, respectively. A criterion is then used to approve if the reassignment is valid by checking the backtracking condition
	\begin{equation} \label{eq:AssignmentCriterion}
		h_l^{\ast, (n)} \geq  h_l^{\ast, (n-1)},
	\end{equation}
	which ensures a nondecreasing objective function in problem~\eqref{Problem:PilotAssignment}. We stress that the condition~\eqref{eq:AssignmentCriterion} needs to be implemented in each iteration due to the non-convexity of \eqref{eq:hln1} and \eqref{eq:hln2}. Moreover the limited power budgets in \eqref{eq:ULPowerConst} and \eqref{eq:DLPowerConst} ensure that this objective function is bounded from above for any set of pilot and data power coefficients in the feasible domain. Consequently, problem \eqref{Problem:PilotAssignment} converges to a fixed point and we conclude the proof.
\end{proof}
During assigning the pilot signals to the users over cells, the proposed approach will be stopped when, for example, a small variation of two consecutive iterations, which is computed for all the $L$ cells as
\begin{equation} \label{eq:StoppingCriterion}
	\sum_{l=1}^L \left| h_l^{\ast, (n)} - h_l^{\ast, (n-1)} \right| \leq \epsilon
\end{equation}
where $\epsilon \geq 0$ is a given accuracy. The proposed pilot assignment approach is applied to all the cells as in Algorithm~\ref{Algorithmv1}.
	
\subsection{Data power control} \label{Subsec:ULDLPowerControl}
For a given pilot assignment, problem \eqref{Problem:DataOptimizationv1} now reduces to the data power control problem. To obtain a low-complexity, the uplink and downlink data power controls can be separately optimized. We therefore present a framework which is applied for both using the nominal parameters:  Let us denote $ \alpha  \in \{ \mathrm{ul}, \mathrm{dl} \}$, the max-min fairness problem is now formulated as
\begin{equation} \label{Problem:DataOptimizationv2}
	\begin{aligned} 
		& \underset{ \{ p_{l,k}^{\alpha} \geq 0 \} }{\mathrm{maximize}} && \underset{(l,k)}{\mathrm{min}}
		\quad  w_{l,k}^{\alpha} R_{l,k}^{\alpha} \\
		& \textrm{subject to}
		& & \mbox{Constraints in } \eqref{eq:Powerbudget},
	\end{aligned}
\end{equation}
where the power budget constraints are
\begin{equation} \label{eq:Powerbudget}
	\begin{cases}
		p_{l,k}^{\mathrm{ul}} \leq P_{\mathrm{max},l,k}^{\mathrm{ul}} \;, \forall l,k, &\mbox{for the uplink}, \\
		\sum_{k=1}^K p_{l,k}^{\mathrm{dl}} \leq P_{\mathrm{max},l,k}^{\mathrm{dl}} \;, \forall l, &\mbox{for the downlink}.
	\end{cases}
\end{equation}
By adopting the epigraph representation \cite{Boyd2004a}, \eqref{Problem:DataOptimizationv2} is equivalent to as
\begin{equation} \label{Problem:DataOptimizationv5}
	\begin{aligned} 
		& \underset{\{ p_{l,k}^{\alpha} \geq 0 \}, \xi }{\mathrm{maximize}} &&  \xi \\
		& \textrm{subject to} && \mathrm{SINR}_{l,k}^{\alpha} \geq \xi \;, \forall l,k, \\
		&&& \mbox{Constraints in } \eqref{eq:Powerbudget},
	\end{aligned}
\end{equation}
where the new optimization variable $\xi \in \{ \xi^{\mathrm{ul}}, \xi^{\mathrm{dl}} \}$ is the minimum SINR per user. In \eqref{Problem:DataOptimizationv5}, the objective function and uplink power constraints aligns with monomials. The SINR expressions and downlink power constraints can be recast as posynomials. Consequently, \eqref{Problem:DataOptimizationv5} is a geometric program whose global optimum is able to attain by a general purpose optimization toolbox \cite{Chien2018a}.\footnote{We have implemented the pilot assignment and data power control iteratively, but no further improvement is observed.}
	
Our proposal to obtain a local solution to \eqref{Problem:DataOptimizationv1} is summarized in Algorithm~\ref{Algorithmv1}. The computational complexity of the pilot assignment is from sorting the pilot signals and from computing $KL$ inverse matrices. The matrix inversion is more expensive since each BS has many antennas. The computational complexity of the pilot assignment is in the order of $\mathcal{O} \left( \nu N_1 L^2 K M^3 \right)$, where $N_1$ is the number of iterations requires to reach the fixed point and the constant value $\nu$ stands for the effectiveness of computing matrix inverse \cite{krishnamoorthy2013matrix}. Next, the data power control by the interior-point method consumes the computational complexity of the order of  $\mathcal{O} \left( 2N_2^\alpha\max \left\{ 2L^3 K^3, F_1 \right\} \right)$, where $F_1$ is the first and second derivative estimation cost of computing the SINR constraints in \eqref{Problem:DataOptimizationv5}. Consequently, the total computational complexity of Algorithm~\ref{Algorithmv1} is in the order of $\mathcal{O} \left( \nu N_1 L^2 K M^3 +   2N_2^\alpha\max \left\{ 2L^3 K^3, F_1 \right\}  \right)$. By conditioned on the dominated computational complexity from the matrix inverses $\mathbf{F}_{l,k}^{-1}, \forall l,k,$ and the computational complexity of the pilot assignment and either the uplink or downlink data power control is $\mathcal{O} \left( \nu N_1 L^2 K M^3 +   N_2^\alpha\max \left\{ 2L^3 K^3, F_1 \right\}  \right)$. 
	
\begin{algorithm}[t]
	\caption{An approach finding a fixed point to \eqref{Problem:DataOptimizationv1}}\label{Algorithmv1}
	\begin{algorithmic}[1]
		\State \textbf{Input} Set $\{ P_{\max,l,k}^{\mathrm{ul}}$, $P_{\max,l,k}^{\mathrm{dl}} \}, \epsilon, h_l^{\ast,(0)} =0, h_l^{\ast,(1)} =1, \forall l, $ and preliminary pilot assignment $\{ \pmb{\psi}_{l,k}^{\ast} \}$ obtained by randomization; Select initial transmit powers $\{p_{l,k}^{\mathrm{ul}}, p_{l,k}^{\mathrm{dl}} \}$. Set $n=0$.
		\While {\eqref{eq:StoppingCriterion}  unsatisfied}{
			\State Set $n = n + 1$
			\For{$l=1, \ldots, L$}
			\State BS~$l$ computes \eqref{eq:Increasingcost} and \eqref{eq:PilotContOrder}, then assigns the pilot signals as in Fig.~\ref{FigPilotAssigment}
			\State BS~$l$ verifies the backtracking condition~$h_l^{\ast, (n)} \geq h_l^{\ast, (n-1)}$
			\textbf{If} this is satisfied, \textbf{then} update $h_l^{\ast,(n)}$ and broadcast the new pilot assignment $\{\pmb{\psi}_{l,k}^{\ast}\}_{k=1}^K$. \textbf{Otherwise} keep the previous one.
			\State BS~$l$ checks the stopping condition. \textbf{If} not satisfied, \textbf{then} continue by setting $l = l+1$. \textbf{Otherwise} go to step 9.
			\EndFor
			\State{Update the cost $\sum_{l=1}^L \big| h_l^{\ast, (n)} - h_l^{\ast, (n-1)} \big|$ for the stopping criterion \eqref{eq:StoppingCriterion}.}
			\EndWhile}
		\State Solve problem \eqref{Problem:DataOptimizationv5} to obtain the optimal data powers  $\{ p_{l,k}^{\ast, \mathrm{ul}},p_{l,k}^{\ast, \mathrm{dl}} \}, \forall l,k.$
		\State \hspace{-0.5cm}\textbf{Output} $\{ p_{l,k}^{\ast, \mathrm{ul}} $, $p_{l,k}^{\ast, \mathrm{dl}} \}$, and $\{ \pmb{\psi}_{l,k}^{\ast} \}, \forall l,k$.
	\end{algorithmic}
\end{algorithm}
\begin{figure}[t]
	\centering
	\includegraphics[trim=0.6cm 0.1cm 0.6cm 0.5cm, clip=true, width=3.2in]{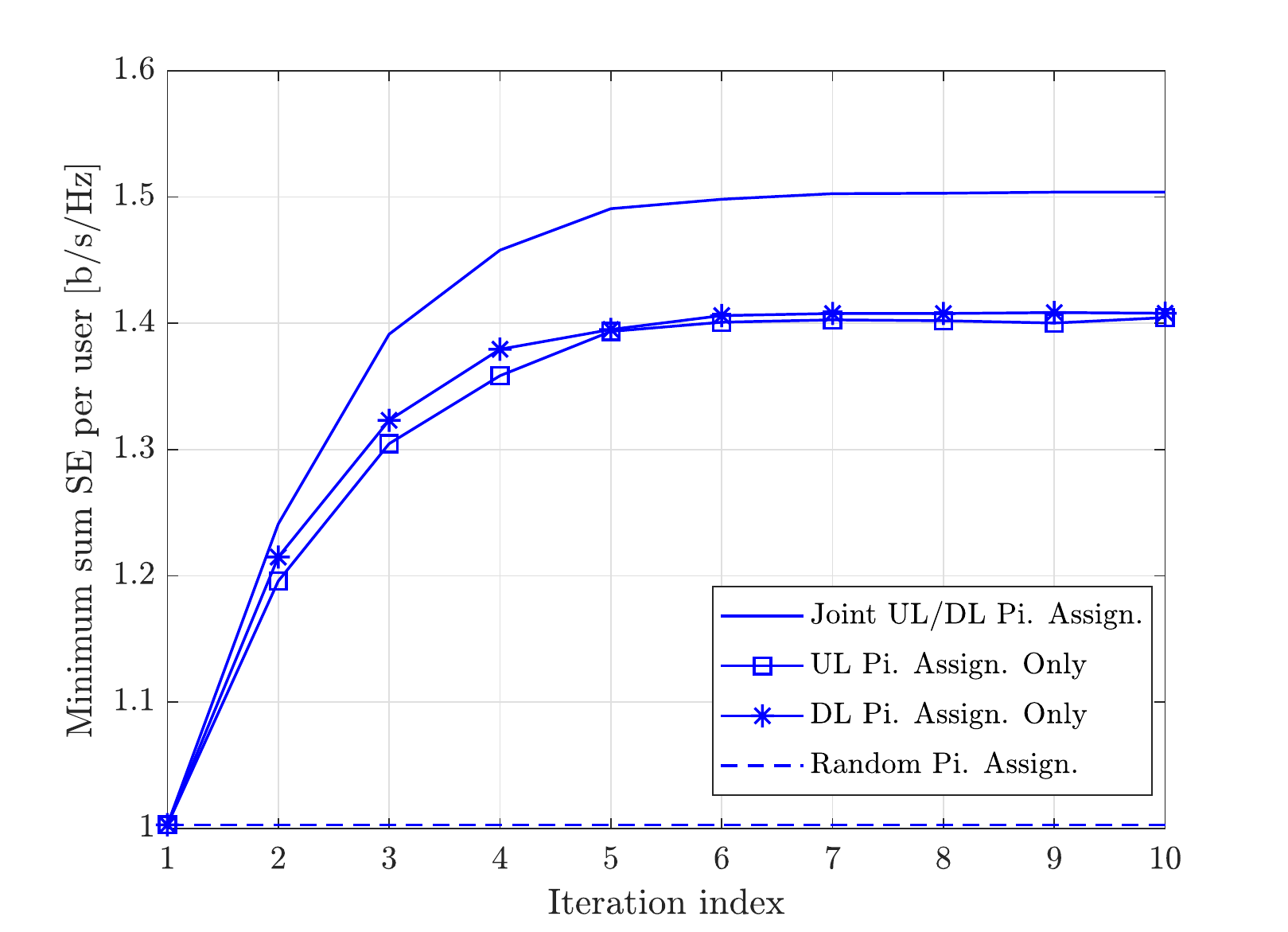} 
	\caption{The convergence of the proposed pilot assignment for a network with $4$ users per cell.}
	\label{FigConvergence}
\end{figure}
	
\section{Numerical Results} \label{Sec:NumericalResults}
A network is considered with $4$ square cells covering the area $0.5$ km$^{2}$ utilizing the wrap-around technique at the edges to avoid boundary effects, and therefore one BS has eight neighbors. In each cell, a BS with $200$ antennas is at the center and serving $K$  uniformly distributed users with a minimum distance to the serving BS being $35$~m. There are $K$ orthogonal pilot signals, while the maximum power is $P_{\mathrm{max},l,k}^{\mathrm{ul}} = 200$ mW and $P_{\mathrm{max},l}^{\mathrm{dl}} = 200K$ mW for an equal total power budget of the uplink and downlink data transmission thanks to the uplink-downlink duality \cite{Boche2002a}. The noise variance is $-96$ dBm corresponding to the noise figure $5$ dB. The large-scale fading coefficients are
\begin{equation}
	\beta_{l,k}^j [\mathrm{dB}] = -148.1 - 37.6 \log_{10}(d_{l,k}^j/1 \textrm{km}) + z_{l,k}^j, 
\end{equation}
where $d_{l,k}^j$ in km is the distance between user~$k$ in cell~$l$ and BS~$j$ \cite{Chien2018a}. The shadow fading $z_{l,k}^j$ follows a log-normal distribution with standard deviation $7$~dB. The covariance matrix of the channel from user~$k$ in cell~$l$ and BS~$j$ is defined by the exponential correlation model, which models a uniform linear array as
\begin{equation}
	\mathbf{R}_{l,k}^j = \beta_{l,k}^j  \begin{bmatrix}
		1 &  r_{l,k}^{j,\ast}  & \cdots & \big( r_{l,k}^{j,\ast} \big)^{M-1} \\ 
		r_{l,k}^{j}& 1 &  \cdots & \big( r_{l,k}^{j,\ast} \big)^{M-2} \\
		\vdots & \vdots & \ddots & \vdots  \\
		\big( r_{l,k}^{j} \big)^{M-1}&  \big( r_{l,k}^{j} \big)^{M-2} & \cdots & 1
	\end{bmatrix},
\end{equation}
where the spatial correlation $r_{l,k}^{j}= \mu e^{j \theta_{l,k}^j}$ with the correlation magnitude $\mu$ in the range $[0,1]$ and the user incidence angle to the array boresight being $\theta_{l,k}^j$. By setting the weights $w_{l,k}^{\mathrm{ul}}, w_{l,k}^{\mathrm{dl}} \in \{ 0, 1 \}, \forall l,k,$ (i.e., $w_{l,k}^{\mathrm{dl}} = 0$ if only considering the uplink data transmission; $w_{l,k}^{\mathrm{ul}} = 0$ if only considering the downlink data transmission; $w_{l,k}^{\mathrm{ul}} = w_{l,k}^{\mathrm{dl}} =  1$ if both the uplink and downlink data transmissions are considered) the following benchmarks are used for comparison:\footnote{Exhaustive research  is not included for comparison due to its extremely heavy complexity. One realization of user locations needs to evaluate the SE of $(K!)^{L-1} = 373,248,000$ combinations of pilot signals.} 
\begin{itemize}
	\item[$1)$] \textit{Random pilot assignment (Denoted as ``Ran. Pi. Assign." in the figures):} The pilot signals are randomly assigned to all users, which was used in, for example \cite{Chien2018a}.
	\item[$2)$] \textit{Greedy pilot assignment (Denoted as ``Gre. Pi. Assign." in the figures)}: The pilot signals are assigned based on the similarity between the covariance matrices, which was proposed by \cite{You2015a}.
	\item[$3)$] \textit{Uplink pilot assignment (Denoted as ``UL. Pi. Only" in the figures):} The pilot signals are assigned based on the uplink SE only.
	\item[$4)$] \textit{Downlink pilot assignment (Denoted as ``DL. Pi. Only" in the figures):} The pilot signals are assigned based on the downlink SE only. 
	\item[$5)$] \textit{Pilot assignment  for the joint UL/DL SE enhancement (Denoted as ``Joint UL/DL Pi. Assign." in the figures):} The pilot signals are assigned by the weighted sum SE per user.
\end{itemize}
\begin{figure}[t]
	\centering
	\includegraphics[trim=0.6cm 0.2cm 0.6cm 0.7cm, clip=true, width=3.2in]{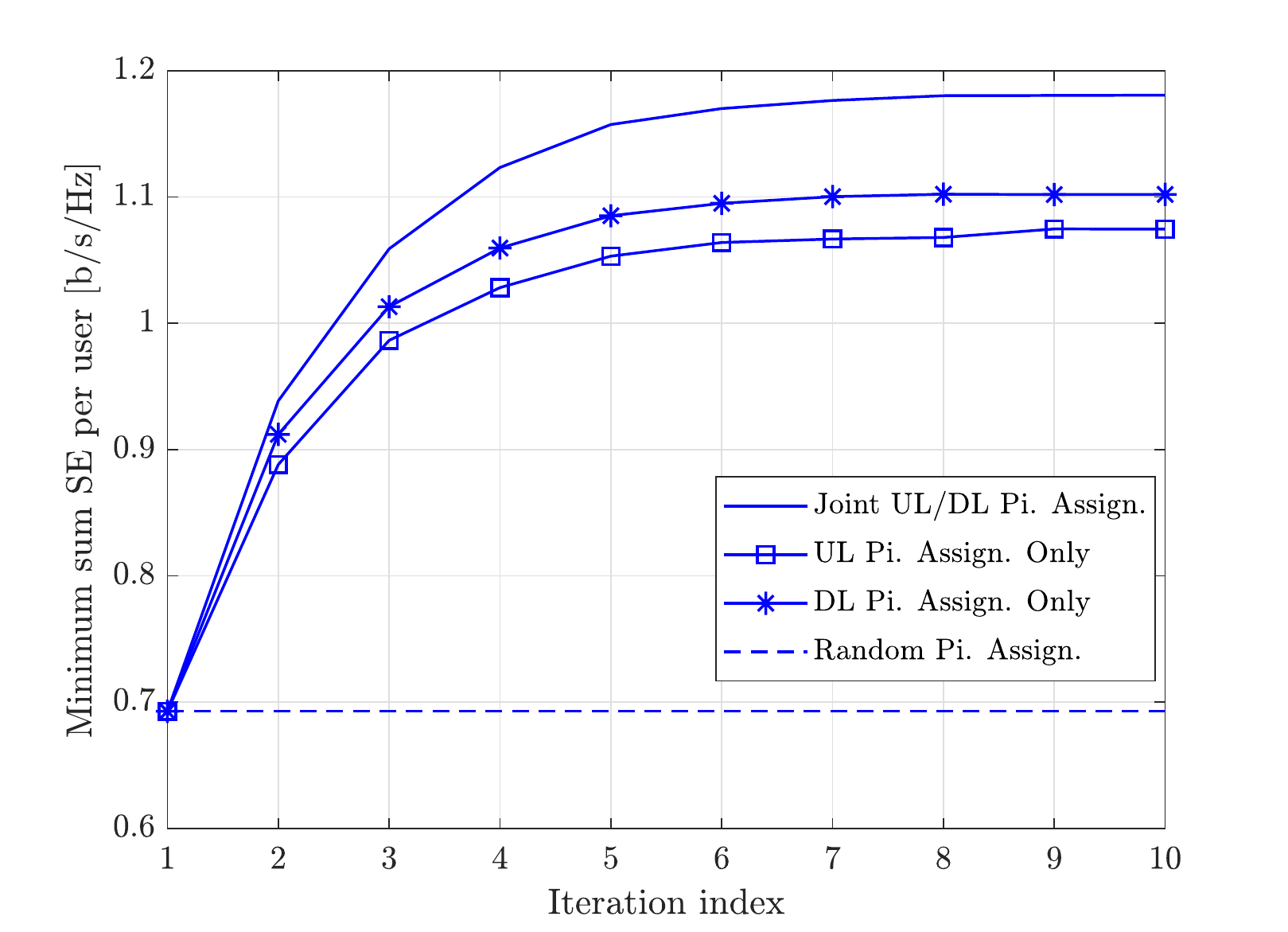} 
	\caption{The convergence of the proposed pilot assignment for a network with $6$ users per cell.}
	\label{FigConvergencev1}
\end{figure}
\begin{figure}[t]
	\centering
	\includegraphics[trim=0.6cm 0.2cm 0.6cm 0.7cm, clip=true, width=3.2in]{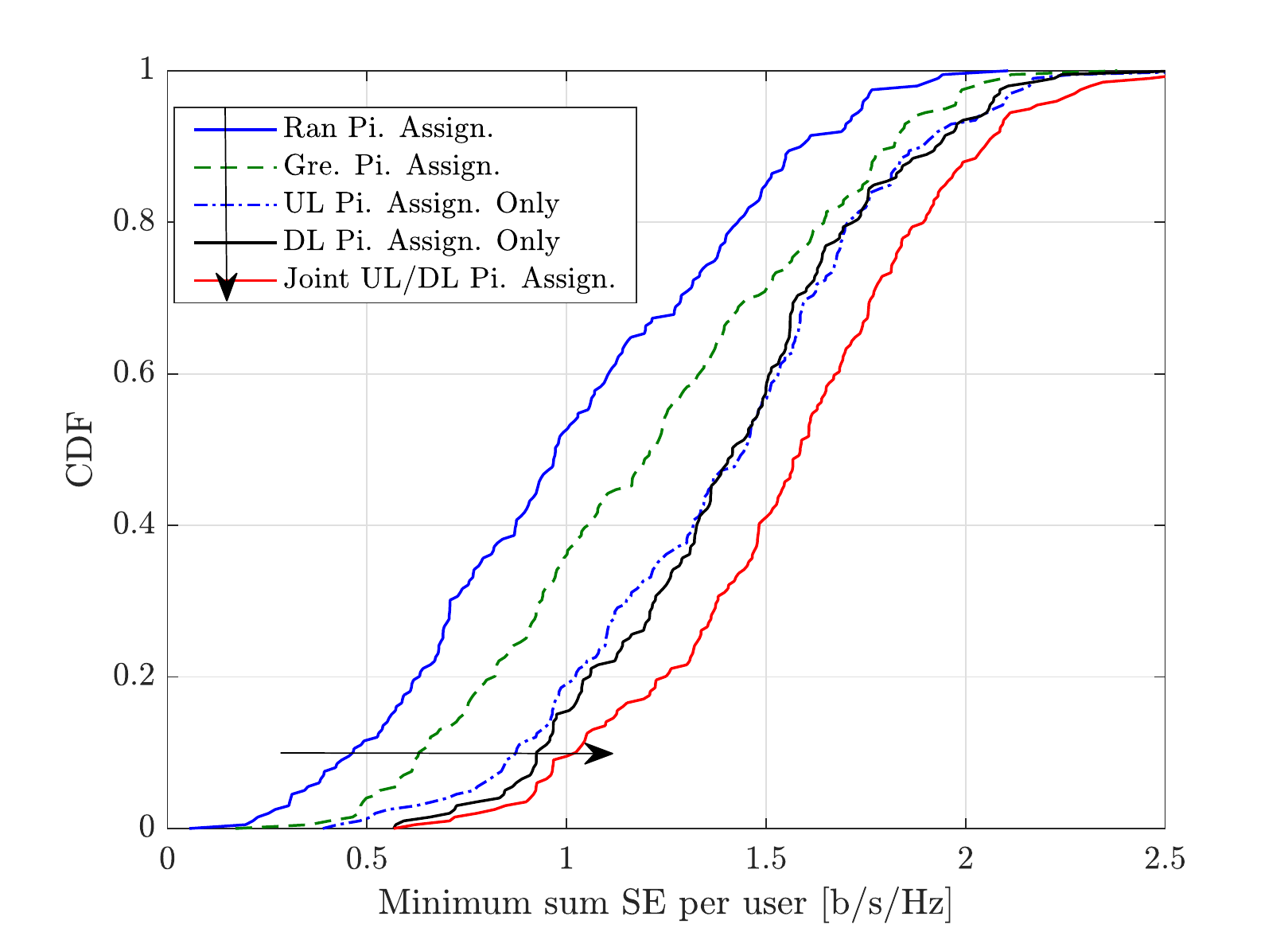} 
	\caption{The cumulative distribution function of the minimum sum SE per user without data power control.}
	\label{FigCDFNoPowerCtr}
\end{figure}
\begin{figure}[t]
	\centering
	\includegraphics[trim=4cm 9.2cm 3.8cm 9.7cm, clip=true, width=3.2in]{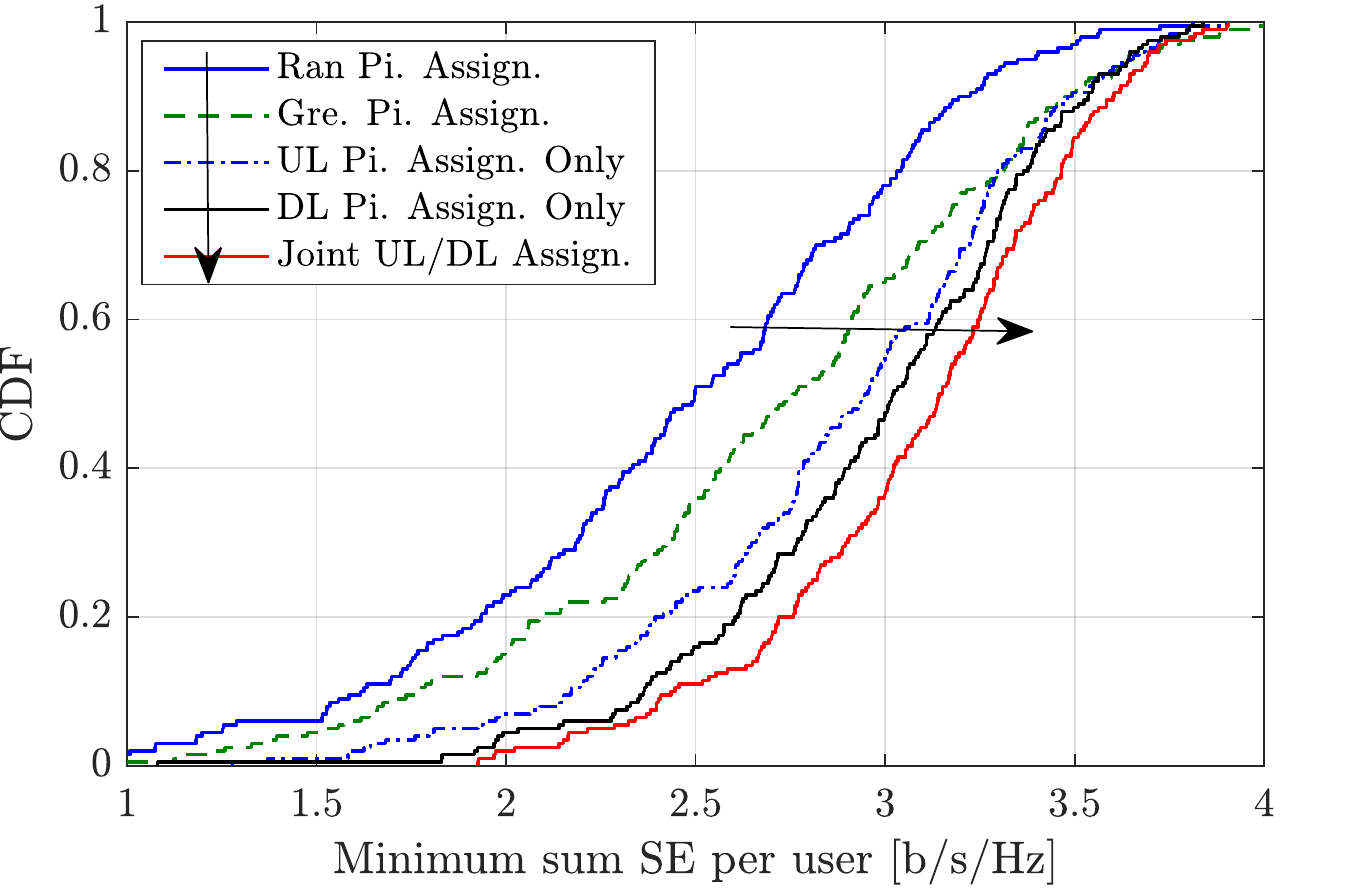} 
	\caption{The cumulative distribution function of the minimum sum SE per user with data power control.}
	\label{FigCDFPowerCtr}
\end{figure}
\begin{figure}[t]
	\centering
	\includegraphics[trim=0.6cm 0.2cm 0.6cm 0.7cm, clip=true, width=3.2in]{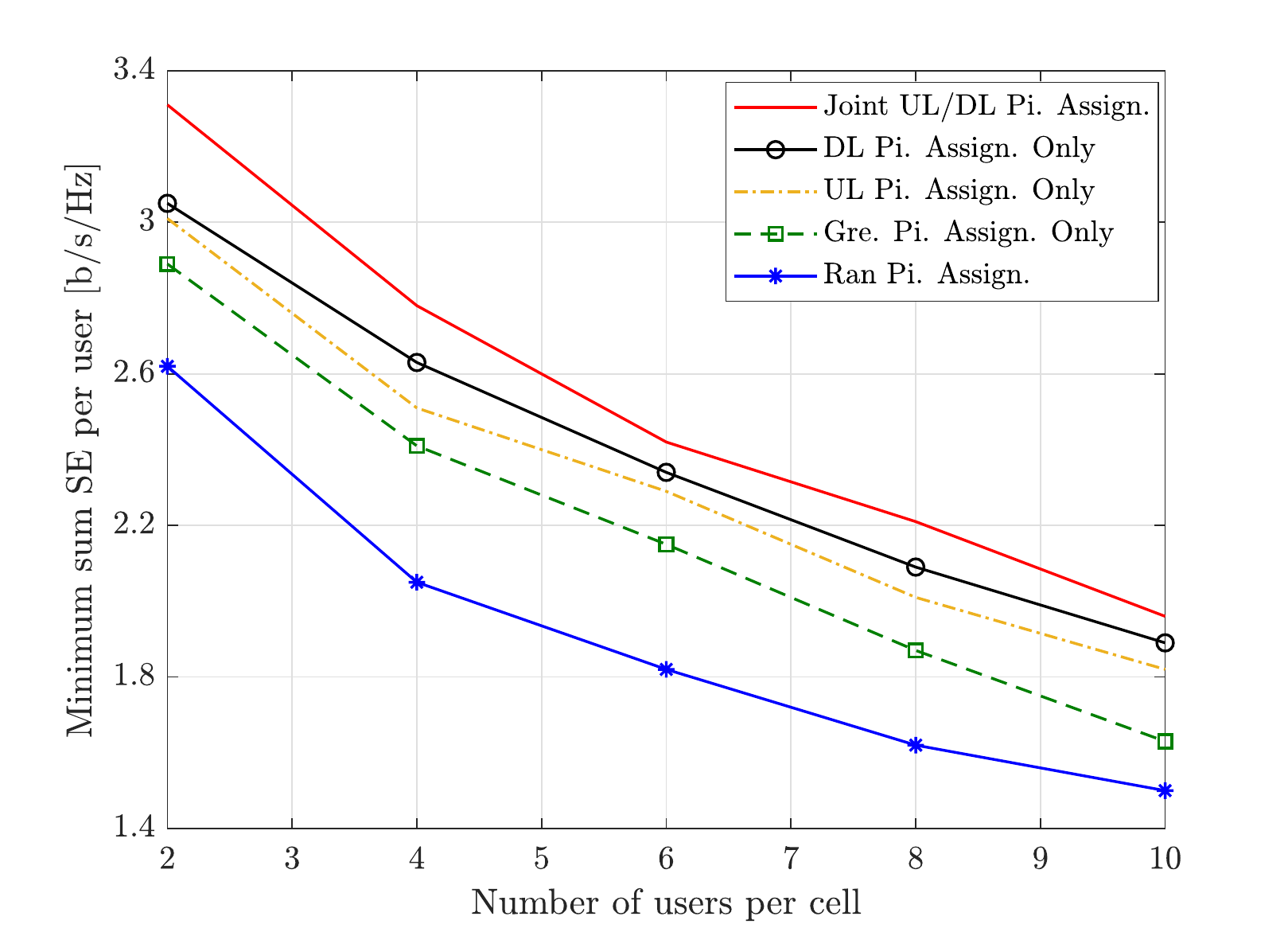} 
	\caption{The minimum sum SE per user versus the number of users per cell with data power control.}
	\label{FigAverageSE}
\end{figure}
Figs.~\ref{FigConvergence} and \ref{FigConvergencev1} display the convergence of proposed pilot assignments for a network with $4$ users and $6$ users per cell, respectively. The convergence is obtained after less than $8$ iterations for all the considered scenarios. When each BS serves $4$ users, the sum SE per user converges to about $1.4$ b/s/Hz when utilizing either the uplink or downlink SE as the utility metric to assign the pilot signals. However, relying on the downlink SE to assign the pilot signals yields $2\%$ better the sum SE per user than utilizing the uplink SE when increasing the number of users per cell to $6$ users.
	
Fig.~\ref{FigCDFNoPowerCtr} shows the cumulative distribution function (CDF) of the minimum weighted SE per user for a network where each cell has $4$ users. The greedy pilot gives $1.2 \times$ SE better than the random assignment. While assigning the pilot signals based on either the uplink or downlink SE gives almost equivalent performance, but it provides better performance than the random pilot assignment by $1.5 \times$. Meanwhile, the improvement of joint pilot assignment is  up to $39\%$ in weighted minimum sum SE per user compared with the second best and it approves the locality of Algorithm~\ref{Algorithmv1}. Finally, Fig.~\ref{FigCDFPowerCtr} manifests the benefits of data power control based on the proposed pilot assignment over the other benchmarks. We observe that the greedy pilot assignment only outperforms the random pilot assignment $1.1\times$. Meanwhile, an improvement up to $2.29 \times$ better weighted sum SE per user than random pilot assignment is obtained. In addition, the data power control improves the sum SE per user up to about $2\times$ compared with the fixed data power allocation.
	
Fig.~\ref{FigAverageSE} plots the minimum sum SE per user versus the number of users per cell. Specifically, the minimum sum SE per user decreases when there are more users in the coverage area that generate more mutual interference. For instance, The joint pilot assignment reduce the minimum sum SE per user $1.7 \times$ as the number of user per cell increases from $2$ to $8$ users. We also observe the benefits of combining the joint pilot assignment and data power control that results in a superior SE improvement up to  $1.4 \times$ compared with the random assignment.

\section{Conclusion}
This paper has formulated and solved a max-min sum SE per user optimization problem considering both the pilot assignment and data power control for cellular Massive MIMO systems with correlated Rayleigh channels. We observed significant improvements of pilot assignment to the minimum sum SE per user compared with the other related works. Interestingly, only deploying the uplink or downlink SE as side information to assign pilot still yields good sum SE to weak users if the max-min fairness optimization is considered.
	
\bibliographystyle{IEEEtran}
\bibliography{IEEEabrv,refs}
\end{document}